\documentclass[runningheads,a4paper]{llncs}

\usepackage{amssymb}
\usepackage{graphicx}

\usepackage{url}

\usepackage{algorithmic}
\usepackage{algorithm}
\usepackage{fancyhdr} 
\usepackage{color}
\usepackage{natbib}
\usepackage{comment}
\usepackage{url}
\usepackage{epsf}
\usepackage{epsfig}
\usepackage{multirow}
\usepackage{tikz}
\usepackage{bm}
\usepackage{amsmath, amssymb}
\usepackage{bbm}

\makeatletter
\renewcommand\bibsection%
{
  \section*{\refname
    \@mkboth{\MakeUppercase{\refname}}{\MakeUppercase{\refname}}}
}
\makeatother

\renewcommand{\paragraph}[1]{\vspace*{1em} \noindent {\sc #1.}}

\newcommand{\halmos}{\hspace*{\fill}\rule{1ex}{1.4ex}}
\makeatletter
\def\newproof#1{\@nprf{#1}}
\def\@nprf#1#2{\expandafter\@ifdefinable\csname #1\endcsname
\global\@namedef{#1}{\@prf{#1}{#2}}\global\@namedef{end#1}{\@endproof}}
\def\@prf#1#2{\@beginproof{#2}{\csname the#1\endcsname}\ignorespaces}
\def\@beginproof#1{\rm \trivlist \item[\hskip \labelsep{\bf #1: }]}
\def\@endproof{\halmos \endtrivlist}

\newproof{prftheorem}{Proof of Theorem~\ref{centraltheoremtechnical}}

\def\operator@font{\mathgroup\symoperators}
\def\Lap{\mathop{\operator@font Lap}}

\newcounter{mycount} 
\newcounter{nmycount}

\begin{document}

\title{Efficient data hashing with structured binary embeddings}

\titlerunning{On the boosting ability of top-down decision tree learning algorithm for multiclass classification}

\author{}
\author{Krzysztof Choromanski \\ Google Research}
\authorrunning{}

\institute{}

\maketitle

\begin{abstract}
We present here new mechanisms for hashing data via binary embeddings. Contrary to most of the techniques presented before, the embedding matrix of our mechanism is highly structured.
That enables us to perform hashing more efficiently and use less memory. What is crucial and nonintuitive is the fact that imposing structured mechanism does not affect the quality of the produced hash. To the best of our knowledge, we are the first to give strong theoretical guarantees of the proposed binary hashing method by proving the efficiency of the mechanism for several classes of structured projection matrices. As a corollary, we obtain binary hashing mechanisms with strong concentration results for circulant and Topelitz matrices. Our approach is however much more general.
\end{abstract}

\section{Hashing mechanism}
\label{sec:has_mech}

In this section we explain in detail proposed hashing mechanism for initial dimensionality reduction
that is used to preprocess data before it is given as an input to the autoencoder. As mentioned earlier, 
the mechanism is of its own interest.
We introduce first the aforementioned family of $\Psi$-regular matrices $\mathcal{P}$ that is
a key ingredient of the method.

Assume that $k$ is the size of the hash and $n$ is the dimensionality of the data.
Let $t$ be the size of the pool of independent random gaussian variables $\{g_{1},...,g_{t}\}$,
where each $g_{i} \sim \mathcal{N}(0,1)$. Assume that $k \leq n \leq t \leq kn$.
We say that a random matrix $\mathcal{P}$ is $\Psi$-regular if $\mathcal{P}$ is of the form:

\begin{equation}
\left( \begin{array}{ccccc}
\sum_{l \in \mathcal{S}_{1,1}}g_{l}     &  ... & \sum_{l \in \mathcal{S}_{1,j}}g_{l} & ... & \sum_{l \in \mathcal{S}_{1,n}}g_{l}\\
... & ... & ... & ... & ... \\
\sum_{l \in \mathcal{S}_{i,1}}g_{l} & ... & \sum_{l \in \mathcal{S}_{i,j}}g_{l} & ... & \sum_{l \in \mathcal{S}_{i,n}}g_{l}\\
... & ... & ... & ... & ... \\
\sum_{l \in \mathcal{S}_{k,1}}g_{l} & ... & \sum_{l \in \mathcal{S}_{k,j}}g_{l} & ... & \sum_{l \in \mathcal{S}_{k,n}}g_{l}
\end{array} \right)
\end{equation}

where $S_{i,j} \subseteq \{1,...,t\}$ for $i \in \{1,...,k\}$, $j \in \{1,...,n\}$, $|S_{i,1}|=...=|S_{i,n}|$ for $i=1,...,k$, $S_{i,j} \cap S_{i,u} = \emptyset$
for $i \in \{1,...,k\}$, $\{j,u\} \subseteq \{1,...,n\}$, $j \neq u$ and furthermore the following holds: 
\begin{itemize}
\item for every column of $\mathcal{P}$ every $g_{l}$ appears in at most $\Phi$ entries from that column.
\end{itemize}

Notice that all structured matrices that we mentioned in the abstract are special cases of the $0$-regular matrix.
Indeed, each Toeplitz matrix is clearly $0$-regular, where subsets $\mathcal{S}_{i,j}$ are singletons.

Let $\phi$ be a function satisfying $\lim_{x \rightarrow \infty} \phi(x) = 1$ and $\lim_{x \rightarrow -\infty} \phi(x) = -1$.
We will consider two hashing methods. The first one, called by us \textit{extended $\Psi$-regular hashing}, applies first random diagonal matrix $\mathcal{R}$ to the datapoint $x$,
then the $L_{2}$-normalized Hadamard matrix $\mathcal{H}$, next another random diagonal matrix $\mathcal{D}$, then the $\Psi$-regular projection matrix $\mathcal{P}_{\Psi}$ and finally function $\phi$ (the latter one applied pointwise).
The overal scheme is presented below:
\begin{equation}
x \xrightarrow {\mathcal{R}} x_{\mathcal{R}} \xrightarrow {\mathcal{H}} x_{\mathcal{H}} \xrightarrow {\mathcal{D}} x_{\mathcal{D}} \xrightarrow {\mathcal{P}_{\Psi}} x_{\mathcal{P}_{\Psi}} \xrightarrow {\phi} h(x) \in \mathbb{R}^{k}.
\end{equation}
The diagonal entries of matrices $\mathcal{R}$ and $\mathcal{D}$ are chosen independently from the binary set $\{-1,1\}$,
each value being chosen with probability $\frac{1}{2}$.
We also propose a shorter pipeline, called by us \textit{short $\Psi$-regular hashing}, where we avoid applying first random matrix and Hadamard matrix $\mathcal{R}$ and the Hadamard matrix, i.e. the overall pipeline is of the form:
\begin{equation}
x  \xrightarrow {\mathcal{D}} x_{\mathcal{D}} \xrightarrow {\mathcal{P}_{\Psi}} x_{\mathcal{P}_{\Psi}} \xrightarrow {\phi} h(x) \in \mathbb{R}^{k}.
\end{equation}
The goal is to compute good approximation of the angular distance between given $L_{2}$-normalized vectors $p,r$, given their 
compact hashed versions: $h(p), h(r)$. 
To achieve this goal we consider the $L_{1}$-distance in the $k$-dimensional space of hashes.
Let $\theta_{p,r}$ denote the angle between vectors $p$ and $r$. We define the \textit{normalized approximate angle between $p$ and $r$} as:  
\begin{equation}
\tilde{\theta}_{p,r}^{n} = \frac{1}{2k}\|h(p)-h(r)\|_{1}
\end{equation} 
In the next section we will show that the normalized approximate angle between vectors $p$ and $r$ is a very precise estimation of the actual angle if the chosen parameter $\Psi$ is not large enough. Furthermore, we show an intriguing connection between theoretical guarantess regarding the quality of the produced hash and the chromatic number of some specific undirected graph encoding the structure of $\mathcal{P}$. For many of the structured matrices under consideration this graph is induced by an algebraic group operation defining the structure of $\mathcal{P}$ (for istance, for the circular matrix the group is a single shift and the underlying graph is a collection of pairwise disjoint cycles and trees thus its chromatic number is at most $3$).

\section{Theoretical results}
\label{sec:the}
\subsection{Introduction}
We are ready to provide theoretical guarantees regarding the quality of the produced hash. Our guarantees will be given for a \textit{sign} function, i.e for $\phi$ defined as: $\phi(x) = 1$ for $x \geq 0$, $\phi(x) = -1$ for $x < 0$. However we should emphasize that empirical results showed that other functions (that are often used as nonlinear maps in deep neural networks) such as sigmoid function, also work well. It is not hard to show that $\tilde{\theta}_{p,r}^{n}$ is an unbiased estimator of $\frac{\theta_{p,r}}{\Pi}$, i.e. $E(\tilde{\theta}_{p,r}^{n}) = \frac{\theta_{p,r}}{\Pi}$. What we will focus on is the concentration of the random variable $\tilde{\theta}_{p,r}^{n}$ around its mean $\frac{\theta_{p,r}}{\Pi}$. We will prove strong exponential concentration results regarding the extended $\Psi$-regular hashing method. Interestingly, the application of the Hadamard mechanism is not necessary and it is possible to get concentration results, yet weaker than in the former case, also for short $\Psi$-regular hashing.
As a warm up, let us prove the following.

\begin{lemma}
\label{mean_lemma}
Let $\mathcal{M}$ be a $\Psi$-regular hashing model (either extended or short). Then $\tilde{\theta}_{p,r}^{n}$ is an unbiased estimator of $\theta_{p,r}$, i.e. 
$$E(\tilde{\theta}_{p,r}^{n}) = \frac{\theta_{p,r}}{\Pi}.$$
\end{lemma}

\begin{proof}
Notice first that the $i$th row, call it $g^{i}$, of the matrix $\mathcal{P}$ is a $n$-dimensional gaussian vector with mean $0$ and where each element has standard deviation $\sigma_{i}$ for $\sigma_{i}=|\mathcal{S}_{i,1}|=...=|\mathcal{S}_{i,n}|$ ($i=1,...,k$).
Thus, after applying matrix $\mathcal{D}$ the new vector $g^{i}_{\mathcal{D}}$ is still gaussian and of the same distribution.
Let us consider first the short $\Psi$-regular hashing model.
Fix some $L_{2}$-normalized vectors $p,r$ (without loss of generality we may assume that they are not collinear) and denote by $H_{p,r}$ the $2$-dimensional hyperplane spanned by $\{p,r\}$. Denote by $g^{i}_{\mathcal{D},H}$ the projection of 
$g^{i}_{\mathcal{D}}$ into $H$ and by $g^{i}_{\mathcal{D},H,\perp}$ the line in $H$ perpendicular to $g^{i}_{\mathcal{D},H}$. Let $\phi$ be a \textit{sign} function. Notice that the contribution to the $L_{1}$-sum $\|h(p)-h(r)\|_{1}$ comes from those $g^{i}$ for which
$g^{i}_{\mathcal{D},H,\perp}$ divides an angel between $p$ and $r$, i.e. from those $g^{i}$ for which $g^{i}_{\mathcal{D},H}$
is inside the union $\mathcal{U}_{p,r}$ of two $2$-dimensional cones bounded by two lines in $H$ perpendicular to $p$ and $r$
respectively. 
Observe that, from what we have just said, we can conclude that $\tilde{\theta}_{p,r}^{n} = \frac{X_{1} + ... + X_{k}}{k}$, where:
\begin{equation}
X_{i} =
\left\{
	\begin{array}{ll}
		1  & \mbox{if }  g^{i}_{\mathcal{D},H} \in \mathcal{U}_{p,r}, \\
		0 & \mbox{otherwise.} 
	\end{array}
\right.
\end{equation}

Now it suffices to notice that vector $g^{i}_{\mathcal{D},H}$ is a gaussian random variable and thus its direction is uniformly distributed over all directions. Thus each $X_{i}$ is nonzero with probability exactly $\frac{\theta}{\Pi}$ and the theorem follows.
For the extended $\Psi$-regular hashing model the analysis is very similar.
The only difference is that data is preprocessed by applying $\mathcal{H}\mathcal{R}$ linear mapping first. Both $\mathcal{H}$ and $\mathcal{R}$ are matrices of rotations though, thus their product is also a rotation matrix. Since rotations do not change angular distance, the former analysis can be applied again and yields the proof.
\end{proof}

\subsection{The $\mathcal{P}$-chromatic number}

As we have already mentioned, the highly well organized structure of the projection matrix $\mathcal{P}$ gives rise
to the underlying undirected graph that encodes dependencies between different entries of $\mathcal{P}$.
More formally, let us fix two rows of $\mathcal{P}$ of indices $1 \leq k_{1} < k_{2} \leq k$.
We define a graph $\mathcal{G}_{\mathcal{P}}(k_{1},k_{2})$ as follows: 
\begin{itemize}
\item $V(\mathcal{G}_{\mathcal{P}}(k_{1},k_{2})) = \{\{j_{1},j_{2}\}: \exists l \in \{1,...,t\} s.t.  g_{l} \in \mathcal{S}_{k_{1},j_{1}} \cap \mathcal{S}_{k_{2},j_{2}}, j_{1} \neq j_{2}\}$,
\item there exists an edge between vertices $\{j_{1},j_{2}\}$ and $\{j_{3},j_{4}\}$ iff $\{j_{1},j_{2}\} \cap \{j_{3},j_{4}\} \neq \emptyset$.
\end{itemize}

The chromatic number $\chi(\mathcal{G})$ of the graph $\mathcal{G}$ is the minimal number of colors that can be used to color the vertices of the graph in such a way that no two adjacent vertices have the same color.

\begin{definition}
Let $\mathcal{P}$ be a $\Psi$-regular matrix. We define the $\mathcal{P}$-chromatic number $\chi(\mathcal{P})$ 
as:
$$\chi(\mathcal{P}) = \max_{1 \leq k_{1} < k_{2} \leq k} \chi(\mathcal{G}(k_{1},k_{2})).$$
\end{definition}

\subsection{Concentration inequalities for structured hashing with \textit{sign} function}

We present now our main theoretical results. Let us consider first the extended $\Psi$-regular hashing model. 
The following is true.

\begin{theorem}
\label{ext_technical_theorem}
Take the extended $\Psi$-regular hashing model $\mathcal{M}$ with $t$ independent gaussian random variables: $g_{1},...,g_{t}$, each of distribution $\mathcal{N}(0,1)$.
Let $N$ be the size of the dataset. Denote by $k$ the size of the hash and by $n$ the dimensionality of the data. Let $f(n)$ be arbitrary positive function. 
Let $p, r$ be two fixed vectors $p,r \in \mathbb{R}^{n}$ with angular distance $\theta_{p,r}$ between them.
Then for every $a,\epsilon>0$ the following is true:
$$
\mathbb{P}(|\tilde{\theta}^{n}_{p,r} - \frac{\theta}{\Pi}| \leq \epsilon) \geq (1-4{N \choose 2}e^{-\frac{f^{2}(n)}{2}}-4\chi(\mathcal{P}){k \choose 2}e^{-\frac{2a^{2}t}{f^{4}(t)}})(1-\Lambda),
$$
where $\Lambda = \frac{1}{\Pi} \sum_{j = \frac{\epsilon k}{2}}^{k}\frac{1}{\sqrt{j}}(\frac{ke}{j})^{j}\mu^{j}(1-\mu)^{k-j}+2e^{-\frac{\epsilon^{2}k}{2}}$ and $\mu=\frac{8k(a\chi(\mathcal{P}) + \Psi\frac{f^{2}(n)}{n})}{\theta_{p,r}}$.
\end{theorem}

Notice how the upper bound on the probability of failure $\mathbb{P}_{\epsilon}$ depends on the $\mathcal{P}$-chromatic number. The theorem above guarantees strong concentration of $\tilde{\theta}^{n}_{p,r}$ around its mean and therefore justifies theoretically the effectiveness of the structured hashing method. It becomes more clearly below.

As a corollary, we obtain the following result:

\begin{theorem}
\label{ext_theorem}
Take the extended $\Psi$-regular hashing model $\mathcal{M}$ with. Assume that the projection matrix $\mathcal{P}$ is Toeplitz.
Let $N$ be the size of the dataset. Denote by $k$ the size of the hash and by $n$ the dimensionality of the data. Let $f(n)$ be an arbitrary positive function. 
Let $p, r$ be two vectors $p,r \in \mathbb{R}^{n}$ with angular distance $\theta_{p,r}$ between them.
Then for every $\epsilon>0$ the following is true: 
$$\mathbb{P}(|\tilde{\theta}^{n}_{p,r} - \frac{\theta}{\Pi}| \leq k^{-\frac{1}{3}}) \geq (1-O(\frac{N^{2}}{n^{4.5}})-O(k^{2}
e^{-\Omega(\frac{n^{\frac{1}{3}}}{\log^{2}(n)})}))(1-(\frac{k^{7}}{n})^{\frac{1}{3}}).$$
\end{theorem}

Theorem \ref{ext_theorem} follows from Theorem \ref{ext_technical_theorem} by taking:
$a=n^{-\frac{1}{3}}$, $\epsilon = k^{-\frac{1}{3}}$, $f(n)=3\sqrt{\log(n)}$ and noticing that every Toeplitz matrix is $0$-regular and the corresponding $\mathcal{P}$-chromatic number $\chi(\mathcal{P})$ is at most $3$.

Let us switch now to the short $\Psi$-regular hashing model. The theorem presented below is the application of the 
Chebyshev's inequality preceded by the careful analysis of the variance $Var(\tilde{\theta}^{n}_{p,r})$.

\begin{theorem}
\label{short_theorem}
Take the short $\Psi$-regular hashing model $\mathcal{M}$, where $\mathcal{P}$ is a Toeplitz matrix.
Let $N$ be the size of the dataset. Denote by $k$ the size of the hash and by $n$ the dimensionality of the data. 
Let $p, r$ be two vectors $p,r \in \mathbb{R}^{n}$ with angular distance $\theta_{p,r}$ between them.
Then the following is true for any $c>0$:
$$\mathbb{P}(|\tilde{\theta}^{n}_{p,r} - \frac{\theta}{\Pi}| \geq c (\frac{\sqrt{\log(k)}}{k})^{\frac{1}{3}}) = O(\frac{1}{c^{2}}).$$
\end{theorem}

The proofs of Theorem \ref{ext_technical_theorem} and Theorem \ref{short_theorem} will be given in the Appendix.

\section{Appendix}

In this section we prove Theorem \ref{ext_technical_theorem} and Theorem \ref{short_theorem}.
We will use notation from Lemma \ref{mean_lemma}.

\subsection{Proof of Theorem \ref{ext_technical_theorem}}

We start with the following technical lemma:

\begin{lemma}
\label{first_lemma}
Let $\{Z_{1},...,Z_{k}\}$ be the set of $k$ independent random variables defined on $\Omega$ such that each $Z_{i}$ has the same distribution and $Z_{i} \in \{0,1\}$. Let $\{\mathcal{F}_{1},...,\mathcal{F}_{k}\}$ be the set of events, where each $\mathcal{F}_{i}$ is in the $\sigma$-field defined by $Z_{i}$ (in particular $\mathcal{F}_{i}$ does not depend on the $\sigma field$ $\sigma(Z_{1},...,Z_{i-1},Z_{i+1},...Z_{k})$). Assume that there exists $\mu < \frac{1}{2}$ such that: $\mathbb{P}(\mathcal{F}_{i}) \leq \mu$ for $i=1,...,k$.
Let $\{U_{1},...,U_{k}\}$ be the set of $k$ random variables such that $U_{i} \in \{0,1\}$ and
$U_{i} | \mathcal{F}_{i} = Z_{i} |\mathcal{F}_{i}$ for $i=1,...,k$, where $X|\mathcal{F}$ stands
for the random variable $X$ truncated to the event $\mathcal{F}$. Assume furthermore that $E(U_{i})=E(Z_{i})$ for $i=1,...,k$.
Denote $Y = \frac{Y_{1}+...+Y_{k}}{k}$. Then the following is true.
\begin{equation}
\mathbb{P}(|Y-EY| > a) \leq \frac{1}{\Pi} \sum_{r=\frac{ak}{2}}^{k}\frac{1}{\sqrt{r}}(\frac{ke}{r})^{r}\mu^{r}(1-\mu)^{k-r} + 2e^{-\frac{a^{2}k}{2}}.
\end{equation}
\end{lemma}

\begin{proof}
Let us consider the event $\mathcal{F}_{bad}$ = $\mathcal{F}_{1} \cup ... \cup \mathcal{F}_{k}$.
Notice that $\mathcal{F}_{bad}$ may be represented by the union of the so-called $r$-blocks, i.e.
\begin{equation}\mathcal{F}_{bad} = \bigcup_{Q \subseteq \{1,...,k\}} (\bigcap_{q \in Q} \mathcal{F}_{q} \bigcap_{q \in \{1,...,k\} \setminus Q} \mathcal{F}^{c}_{q}), \end{equation}

where $\mathcal{F}^{c}$ stands for the complement of event $\mathcal{F}$.
Let us fix now some $Q \subseteq \{1,...,k\}$. Denote \begin{equation}\mathcal{F}_{Q} = \bigcap_{q \in Q} \mathcal{F}_{q} \bigcap_{q \in \{1,...,k\} \setminus Q} \mathcal{F}^{c}_{q}. \end{equation}
Notice that $\mathbb{P}(\mathcal{F}_{Q}) \leq \mu^{r}(1-\mu)^{k-r}$. It follows directly from the Bernoulli scheme. 

Denote $X = \frac{X_{1}+...+X_{k}}{k}$. 
From what we have just said and from the definition of $\{\mathcal{F}_{1},...,\mathcal{F}_{k}\}$ we conclude that for any given $c$ the following holds:
\begin{equation}
\label{xy_diff}
\mathbb{P}(|Y-X| > c) \leq \sum_{r=ck}^{k}{k \choose r}\mu^{r}(1-\mu)^{k-r}.
\end{equation}

Notice also that from the assumptions of the lemma we trivially get: $E(Y)=E(X)$.

Let us consider now the expression $\mathbb{P}(|Y-E(Y)|) > a$. 

We get: $\mathbb{P}(|Y-E(Y)|>a) = 
\mathbb{P}(|Y-E(X)| > a) = \mathbb{P}(|Y-X + X-E(X)| > a) \leq \mathbb{P}(|Y-X|+|X-E(X)|>a) \leq \mathbb{P}(|Y-X| > \frac{a}{2}) + \mathbb{P}(|X-E(X)| > \frac{a}{2})$.

From \ref{xy_diff} we get:

\begin{equation}
\mathbb{P}(|Y-X| > \frac{a}{2}) \leq \sum_{r=\frac{ak}{2}}^{k}{k \choose r} \mu^{r}(1-\mu)^{k-r}.
\end{equation}

Let us consider now the expression: \begin{equation}\xi = \sum_{r=\frac{ak}{2}}^{k}{k \choose r} \mu^{r}(1-\mu)^{k-r}.\end{equation}
We have:
\begin{equation}
\xi \leq \sum_{r=\frac{ak}{2}}^{k} \frac{(k-r+1)...(k)}{r!} \mu^{r}(1-\mu)^{k-r}
\leq \sum_{r=\frac{ak}{2}}^{k} \frac{k^{r}}{r!} \mu^{r}(1-\mu)^{k-r}
\end{equation}
From the Stirling's formula we get: $r! = \frac{2\Pi r^{r+\frac{1}{2}}}{e^{r}}(1+o_{r}(1))$.
Thus we obtain:
\begin{equation}
\label{xi_ineq}
\xi \leq (1+o_{r}(1))\sum_{r=\frac{ak}{2}}^{k}\frac{k^{r}e^{r}}{2\Pi r^{r+\frac{1}{2}}}\mu^{r}(1-\mu)^{k-r} \leq \frac{1}{\Pi} \sum_{r=\frac{ak}{2}}^{k}\frac{1}{\sqrt{r}}(\frac{ke}{r})^{r}\mu^{r}(1-\mu)^{k-r}
\end{equation}
for $r$ large enough.

Now we will use the following version of standard Azuma's inequality:
\begin{lemma}
\label{azuma_general}
Let $W_{1},...,W_{k}$ be $k$ independent random variables such that $E(W_{1})=...E(W_{k})=0$.
Assume that $-\alpha_{i} \leq W_{i+1} - W_{i} \leq \beta_{i}$ for $i=2,...,k-1$.
Then the following is true:
$$
\mathbb{P}(|\sum_{i=1}^{k} W_{i}|>a) \leq 2e^{-\frac{2a^{2}}{\sum_{i=1}^{k}(\alpha_{i}+\beta_{i})^{2}}}
$$
\end{lemma}

Now, using Lemma \ref{azuma_general} for $W_{i} = X_{i} - E(X_{i})$ and $\alpha_{i} = E(X_{i}), \beta_{i}=1-E(X_{i})$ we obtain:
\begin{equation}
\label{azuma_simple}
\mathbb{P}(|X-EX| > \frac{a}{2}) \leq 2e^{-\frac{a^{2}k}{2}}.
\end{equation}

Combining \ref{xi_ineq} and \ref{azuma_simple}, we obtain the statement of the lemma.

\end{proof}

Our next lemma explains the role the Hadamard matrix plays in the entire extended $\Psi$-regular hashing mechanism.

\begin{lemma}
\label{hadamard_lemma}
Let $n$ denote data dimensionality and let $f(n)$ be an arbitrary positive function.
Let $D$ be the set of all $L_{2}$-normalized datapoints, where no two datapoints are identical.
Assume that $|D|=N$.
Consider the ${N \choose 2}$ hyperplanes $H_{p,r}$ spanned by pairs of different vectors $\{p,r\}$ from $D$. Then after applying linear transformation $\mathcal{H}\mathcal{R}$ each hyperplane $H_{p,r}$ is transformed into another hyperplane $H^{\mathcal{H}\mathcal{R}}_{p,r}$. Furthermore, the probability $\mathcal{P}_{\mathcal{H}\mathcal{R}} $that for every $H^{\mathcal{H}\mathcal{R}}_{p,r}$ there exist two orthonormal vectors $x=(x_{1},...,x_{n}),y=(y_{1},...,y_{n})$ in $H^{\mathcal{H}\mathcal{R}}_{p,r}$ such that: $|x_{i}|,|y_{i}| \leq \frac{f(n)}{\sqrt{n}}$
satisfies: $$\mathcal{P}_{\mathcal{H}\mathcal{R}}  \geq 1-4{N \choose 2}e^{-\frac{f^{2}(n)}{2}}.$$
\end{lemma}

\begin{proof}
We have already noticed in the proof of Lemma \ref{mean_lemma} that $\mathcal{H}\mathcal{R}$
is a matrix of the rotation transformation. Thus, as an isometry, it clearly transforms each $2$-dimensional hyperplane into another $2$-dimensional hyperplane.
For every pair $\{p,r\}$ let us consider an arbitrary fixed orthonormal pair $\{u,v\}$ spanning $H_{p,r}$.
Denote $u=(u_{1},...,u_{n})$. Let us denote by $u^{\mathcal{H}\mathcal{R}}$ vector obtained from 
$u$ after applying transformation $\mathcal{H}\mathcal{R}$.
Notice that the $j^{th}$ coordinate of $u^{\mathcal{H}\mathcal{R}}$ is of the form:
\begin{equation}
u^{\mathcal{H}\mathcal{R}}_{j} = u_{1}T_{1}+...+u_{n}T_{n},
\end{equation}
where $T_{1},...,T_{n}$ are independent random variables satisfying:

\begin{equation}
T_{i} =
\left\{
	\begin{array}{ll}
		\frac{1}{\sqrt{n}}  & \mbox{w.p }  \frac{1}{2}, \\
		-\frac{1}{\sqrt{n}} & \mbox{otherwise.} 
	\end{array}
\right.
\end{equation}

The latter comes straightforwardly from the form of the $L_{2}$-normalized Hadamard matrix
(i.e a Hadamard matrix, where each row and column is $L_{2}$-normalized).

But then, from Lemma \ref{azuma_general}, and the fact that $\|u\|_{2}=1$, we get for any $a>0$:

\begin{equation}
\mathbb{P}(|u_{1}T_{1}+...+u_{n}T_{n}| \geq a) \leq 2e^{-\frac{2a^{2}}{\sum_{i=1}^{n}(2u_{i})^{2}}} \leq 2e^{-\frac{a^{2}}{2}}.
\end{equation}

Similar analysis is correct for $v^{\mathcal{H}\mathcal{R}}$.
Notice that $v^{\mathcal{H}\mathcal{R}}$ is orthogonal to $u^{\mathcal{H}\mathcal{R}}$
since $v$ and $u$ are orthogonal. Furthermore, both $v^{\mathcal{H}\mathcal{R}}$ and 
$u^{\mathcal{H}\mathcal{R}}$ are $L_{2}$-normalized. Thus $\{u^{\mathcal{H}\mathcal{R}},v^{\mathcal{H}\mathcal{R}}\}$ is an orthonormal pair.

To complete the proof, it suffices to take $a=f(n)$ and apply the union bound over all
vectors $u^{\mathcal{H}\mathcal{R}}$, $v^{\mathcal{H}\mathcal{R}}$ for all  ${N \choose 2}$
hyperplanes.
\end{proof}

From the lemma above we see that applying Hadamard matrix enables us to assume with high probability
that for every hyperplane $H_{p,r}$ there exists an orthonormal basis consisting of vectors with elements of absolute values at most $\frac{f(n)}{\sqrt{n}}$. We call this event $\mathcal{E}_{f}$. Notice that whether $\mathcal{E}_{f}$ holds or not is determined only by $\mathcal{H}$, $\mathcal{R}$ and the initial dataset $D$.

Let us proceed with the proof of Theorem \ref{ext_technical_theorem}.
Let us assume that event $\mathcal{E}_{f}$ holds. Without loss of generality we may assume that 
we have the short $\Psi$-regular hashing mechanism with an extra property that every $H_{p,r}$ has an orthonormal basis consisting of vectors with elements of absolute value at most $\frac{f(n)}{\sqrt{n}}$.
Fix two vectors $p,r$ from the dataset $D$. Denote by $\{x,y\}$ the orthonormal basis of $H_{p,r}$ with the above property. Let us fix the $i$th row of $\mathcal{P}$ and denote it as $(p_{i,1},...,p_{i,n})$.
After being multiplied by the diagonal matrix $\mathcal{D}$ we obtain another vector:
\begin{equation}
w=(\mathcal{P}_{i,1}d_{1},...,\mathcal{P}_{i,n}d_{n}),
\end{equation}

where:

\begin{equation}
\mathcal{D}_{i,j} =
 \begin{pmatrix}
  d_{1} & 0 & \cdots & 0 \\
  0 & d_{2} & \cdots & 0 \\
  \vdots  & \vdots  & \ddots & \vdots  \\
  0 & 0 & \cdots & d_{n}
 \end{pmatrix}.
\end{equation}

We have already noticed that in the proof of Lemma \ref{mean_lemma} that it is the projection of $w$ into $H_{p,r}$ that determines whether the value of the associated random variable $X_{i}$ is $0$ or $1$.
To be more specific, we showed that $X_{i}=1$ iff the projection is in the region $\mathcal{U}_{p,r}$. 
Let us write down the coordinates of the projection of $w$ into $H_{p,r}$ in the $\{x,y\}$-coordinate system.
The coordinates are the dot-products of $w$ with $x$ and $y$ respectively thus in the $\{x,y\}$-coordinate system we can write $w$ as:
\begin{equation}
\label{coord_eq}
w_{\{x,y\}}=(\mathcal{P}_{i,1}d_{1}x_{1},...,\mathcal{P}_{i,n}d_{n}x_{n},\mathcal{P}_{i,1}d_{1}y_{1},...,\mathcal{P}_{i,n}d_{n}y_{n}).
\end{equation} 

Notice that both coordinates are gaussian random variables and they are independent since they were constructed by projecting a gaussian vector into two orthogonal vectors.
Now notice that from our assumption about the structure of $\mathcal{P}$ we can conclude that
both coordinates may be represented as sums of weighted gaussian random variables $g_{i}$ for $i=1,...,t$, i.e.:
\begin{equation}
w_{\{x,y\}}=(g_{1}s_{i,1}+...+g_{t}s_{i,t},g_{1}v_{i,1}+...+g_{t}v_{i,t}),
\end{equation}

where each $s_{i,j}, v_{i,j}$ is of the form $d_{z}x_{z}$ or $d_{z}y_{z}$ for some $z$ that 
depends only on $i,j$. 
Notice also that 
\begin{equation}
s_{i,1}^{2}+...+s_{i,t}^{2} =v_{i,1}^{2}+...+v_{i,t}^{2}.
\end{equation}
The latter inequality comes from the fact that, by \ref{coord_eq}, both coordinates of 
$w_{\{x,y\}}$ have the same distribution.

Let us denote $s_{i}=(s_{i,1},...,s_{i,t})$, $v_{i}=(v_{i,1},...,v_{i,t})$ for $i=1,...,k$.
We need the following lemma stating that with high probability vectors $s_{1},...,s_{k},v_{1},...,v_{k}$
are close to be pairwise orthogonal.

\begin{lemma}
\label{small_dot_product_lemma}
Let us assume that $\mathcal{E}_{f}$ holds. Let $f(n)$ be an arbitrary positive function. Then for every $a>0$ with probability at least 
$\mathbb{P}_{succ} \geq 1 - 4{k \choose 2} e^{-\frac{2a^{2}n}{f^{4}(n)}}$, taken under coin tosses used to construct $\mathcal{D}$, the following is true for every $1 \leq i_{1} \neq i_{2} \leq k$:
\label{pseudo_ortho_lemma}
$$|\sum_{u=1}^{n} s_{i_{1},u}v_{i_{1},u}| \leq a\chi(\mathcal{P}) + \Psi \frac{f^{2}(n)}{n},$$
$$|\sum_{u=1}^{n} s_{i_{1},u}s_{i_{2},u}| \leq a\chi(\mathcal{P}) + \Psi \frac{f^{2}(n)}{n},$$
$$|\sum_{u=1}^{n} v_{i_{1},u}v_{i_{2},u}| \leq a\chi(\mathcal{P}) + \Psi \frac{f^{2}(n)}{n},$$
$$|\sum_{u=1}^{n} s_{i_{1},u}v_{i_{2},u}| \leq a\chi(\mathcal{P}) + \Psi \frac{f^{2}(n)}{n}.$$
\end{lemma}

\begin{proof}
Notice that the we get the first inequality for free from the fact that $x$ is orthogonal to $y$
(in other words, $\sum_{u=1}^{n} s_{i_{1},u}v_{i_{1},u}$ can be represented as $C\sum_{u=1}^{n} x_{i}y_{i}$ and the latter expression is clearly $0$). 
Let us consider now one of the three remaining expressions. Notice that they can be rewritten as:
\begin{equation}E = \sum_{i=1}^{n} d_{\rho(i)}d_{\lambda(i)} x_{\zeta(i)}x_{\gamma(i)}\end{equation} 
or \begin{equation}E = \sum_{i=1}^{n} d_{\rho(i)}d_{\lambda(i)} y_{\zeta(i)}y_{\gamma(i)}\end{equation}
or \begin{equation}E = \sum_{i=1}^{n} d_{\rho(i)}d_{\lambda(i)} x_{\zeta(i)}y_{\gamma(i)}\end{equation} for some
$\rho, \lambda, \zeta, \gamma$. 
Notice also that from the $\Psi$-regularity condition we immediately obtain that $\rho(i)=\lambda(i)$
for at most $\Psi$ elements of each sum. Get rid of these elements from each sum and consider the remaining ones. From the definition of the $\mathcal{P}$-chromatic number, those remaining ones can be partitioned into at most $\chi(\mathcal{P})$ parts, each consisting of elements that are independent random variables (since in the corresponding graph there are no edges between them).
Thus, for the sum corresponding to each part one can apply Lemma \ref{azuma_general}.
Thus one can conclude that the sum differs from its expectation (which clearly is zero since $E(d_{i}d_{j})=0$ for $i \neq j$) by a with probability at most 
\begin{equation}
\mathbb{P}_{a} \leq 2e^{-\frac{2a^{2}}{\sum_{i=1}^{n} x_{\zeta(i)}x_{\gamma(i)}}}
\end{equation}
or
\begin{equation}
\mathbb{P}_{a} \leq 2e^{-\frac{2a^{2}}{\sum_{i=1}^{n} y_{\zeta(i)}y_{\gamma(i)}}}
\end{equation}
or
\begin{equation}
\mathbb{P}_{a} \leq 2e^{-\frac{2a^{2}}{\sum_{i=1}^{n} x_{\zeta(i)}y_{\gamma(i)}}}
\end{equation}

Now it is time to use the fact that event $\mathcal{E}_{f}$ holds.
Then we know that: $|x_{i}|,|y_{i}| \leq \frac{f(n)}{\sqrt{n}}$ for $i=1,...,n$.
Substituting this upper bound for $|x_{i}|,|y_{i}|$ in the derived expressions on the probabilities coming from Lemma \ref{azuma_general}, and then taking the union bound, we complete the proof.
\end{proof}

We can finish the proof of Theorem \ref{ext_technical_theorem}.
From Lemma \ref{small_dot_product_lemma} we see that\\ $s_{1},...,s_{k},v_{1},...,v_{k}$ are
close to pairwise orthogonal with high probability. Let us fix some positive function $f(n)>0$ and some
$a>0$. Denote 

\begin{equation}
\Delta = a\chi(\mathcal{P}) + \Psi \frac{f^{2}(n)}{n}.
\end{equation}

Notice that , by Lemma \ref{small_dot_product_lemma} we see that applying Gram-Schmidt process
we can obtain a system of pairwise orthogonal vectors $\tilde{s}_{1},...,\tilde{s}_{k},\tilde{v}_{1},...,\tilde{v}_{k}$ such that 
\begin{equation} \label{ineq1}\|\tilde{v}_{i}-v_{i}\|_{2} \leq k \Delta. \end{equation}
and 
\begin{equation}\label{ineq2}\|\tilde{s}_{i}-s_{i}\|_{2} \leq k \Delta. \end{equation}

Let us consider again $w_{x,y}$.  Replacing $s_{i}$ by $\tilde{s}_{i}$ and $v_{i}$ by $\tilde{v}_{i}$
in the formula on $w_{x,y}$, we obtain another gaussian vector: $\tilde{w}_{x,y}$ for each row $i$ of the matrix $\mathcal{P}$. Notice however that vectors $\tilde{w}_{x,y}$ have one crucial advantage over vectors $w_{x,y}$, namely they are independent. That comes from the fact that $\tilde{s}_{1},...,\tilde{s}_{k}$,$\tilde{v}_{1},...,\tilde{v}_{k}$ are pairwise orthogonal.
Notice also that from \ref{ineq1} and \ref{ineq2} we obtain that the angular distance between 
 $w_{x,y}$ and $\tilde{w}_{x,y}$ is at most $k\Delta$.

Let $Z_{i}$ for $i=1,...k$ be an indicator random variable that is zero if $\tilde{w}_{x,y}$ is inside the region $\mathcal{U}_{p,r}$ and zero otherwise. 
Let $U_{i}$ for $i=1,...k$ be an indicator random variable that is zero if $w_{x,y}$ is inside the region $\mathcal{U}_{p,r}$ and zero otherwise. 
Notice that $\tilde{\theta}^{n}_{p,r} = \frac{U_{1}+...+U_{k}}{k}$.
Furthermore, random variables $Z_{1},...,Z_{k},U_{1},...,U_{k}$ satisfy the assumptions of Lemma \ref{first_lemma} with $\mu \leq \frac{8\epsilon}{\theta}$, where $\epsilon = k\Delta$.
Indeed,  random variables $Z_{i}$ are independent since vectors $\tilde{w}_{x,y}$ are independent.
From what we have said so far we know that each of them takes value one with probability exactly $\frac{\theta}{\Pi}$.
Furthermore $Z_{i} \neq U_{i}$ only if $w_{x,y}$ is inside $\mathcal{U}_{p,r}$ and $\tilde{w}_{x,y}$
is outside $\mathcal{U}_{p,r}$ or vice versa. The latter event implies (thus it is included in the event) that $w_{x,y}$ is near the border of the region $\mathcal{U}_{p,r}$, namely within an angular distance $\frac{\epsilon}{\theta}$ from one of the four semilines defining $\mathcal{U}_{p,r}$. Thus in particular an event $Z_{i} \neq U_{i}$ is contained in the event of probability at most $2 \cdot 4 \cdot \frac{\epsilon}{\theta}$ that depends only on one $w_{x,y}$.

But then we can apply Lemma \ref{first_lemma}. All we need is to assume that the premises of Lemma \ref{small_dot_product_lemma} are satisfied. But this is the case with probability specified in Lemma \ref{hadamard_lemma} and this probability is taken under random coin tosses used to product $\mathcal{H}$ and $\mathcal{R}$, thus independently from the random coin tosses used to produce $\mathcal{D}$. Putting it all together we obtain the statement of Theorem \ref{ext_technical_theorem}.

\subsection{Proof of Theorem \ref{short_theorem}}

We will borrow some notation from the proof of Theorem \ref{ext_technical_theorem}.
Notice however that in this setting no preprocessing with the use of matrices $\mathcal{H}$
and $\mathcal{R}$ is applied.

\begin{lemma}
\label{variance_lemma}
Define $U_{1},...,U_{k}$ as in the proof of Theorem \ref{ext_technical_theorem}.
Assume that the following is true:
$$|\sum_{u=1}^{n} s_{i_{1},u}v_{i_{1},u}| \leq \Delta,$$
$$|\sum_{u=1}^{n} s_{i_{1},u}s_{i_{2},u}| \leq \Delta,$$
$$|\sum_{u=1}^{n} v_{i_{1},u}v_{i_{2},u}| \leq \Delta,$$
$$|\sum_{u=1}^{n} s_{i_{1},u}v_{i_{2},u}| \leq \Delta.$$ 
for some $0<\Delta<1$.
The the following is true for every fixed $1 \leq i < j \leq k$:
$$|\mathbb{P}(U_{i}U_{j}=1) - \mathbb{P}(U_{i}=1)\mathbb{P}(U_{j}=1)| = O(\Delta).$$
\end{lemma}

The lemma follows from the exactly the same analysis that was done in the last section of the proof of Theorem \ref{ext_technical_theorem} thus we leave it to the reader as an exercise.

Notice that we have: 

\begin{equation}
Var(\tilde{\theta}^{n}_{p,r}) = Var(\frac{U_{1}+...+U_{k}}{k}) = 
\frac{1}{k^{2}}(\sum_{i=1}^{k} Var(U_{i}) + \sum_{i \neq j} Cov(U_{i},U_{j})).
\end{equation} 

Since $U_{i}$ is an indicator random variable that takes value one with probability $\frac{\theta}{\Pi}$,
we get:
\begin{equation}
Var(U_{i}) = E(U_{i}^{2}) - E(U_{i})^{2} = \frac{\theta}{\Pi}(1-\frac{\theta}{\Pi}).
\end{equation}

Thus we have:

\begin{equation}
Var(\tilde{\theta}^{n}_{p,r}) = \frac{1}{k}\frac{\theta(\Pi-\theta)}{\Pi^{2}}+\frac{1}{k^{2}}\sum_{i \neq j} Cov(U_{i},U_{j}).
\end{equation}

Notice however that $Cov(U_{i},U_{j})$ is exactly: $\mathbb{P}(U_{i}U_{j}=1) - \mathbb{P}(U_{i}=1)\mathbb{P}(U_{j}=1)$.

Therefore, using Lemma \ref{variance_lemma}, we obtain:

\begin{equation}
Var(\tilde{\theta}^{n}_{p,r})  = \frac{1}{k}\frac{\theta(\Pi-\theta)}{\Pi^{2}} + O(\Delta).
\end{equation}

It suffices to estimate parameter $\Delta$.
We proceed as in the previous proof. 
We only need to be a little bit more cautious since the condition: $|x_{i}|,|y_{i}| \leq \frac{f(n)}{\sqrt{n}}$ cannot be assumed right now.
We select two rows: $i_{1},i_{2}$ of $\mathcal{P}$.
Notice that , again we see that applying Gram-Schmidt process
we can obtain a system of pairwise orthogonal vectors $\tilde{s}_{i_{1}},\tilde{s}_{i_{i}},\tilde{v}_{i_{i}},\tilde{v}_{i_{2}}$ such that 
\begin{equation} \label{ineq1}\|\tilde{v}_{i_{1}}-v_{i_{2}}\|_{2} \leq \Delta. \end{equation}
and 
\begin{equation}\label{ineq2}\|\tilde{s}_{i_{1}}-s_{i_{2}}\|_{2} \leq \Delta. \end{equation}

The fact that right now the above upper bounda are not multiplied by $k$, as it was the case in the previous proof,
plays key role in obtaining nontrivial concentration results even when no Hadamard mechanism is applied.

We consider the related sums:\\
$E_{1} = \sum_{i=1}^{n} d_{\rho(i)}d_{\lambda(i)} x_{\zeta(i)}x_{\gamma(i)},
  E_{2} = \sum_{i=1}^{n} d_{\rho(i)}d_{\lambda(i)} y_{\zeta(i)}y_{\gamma(i)},\\
  E_{3} = \sum_{i=1}^{n} d_{\rho(i)}d_{\lambda(i)} x_{\zeta(i)}y_{\gamma(i)}$ 
as before. We can again partition each sum into at most $\chi(\mathcal{P})$ subchunks, where
this time $\chi(\mathcal{P}) \leq 3$ (since $\mathcal{P}$ is Toeplitz).
The problem is that applying Lemma \ref{azuma_general}, we get bounds that depend on the expressions of the form \begin{equation} \alpha_{x,i} = \sum_{j=1}^{n}x_{j}^{2}x_{j+i}^{2}\end{equation} and \begin{equation} \alpha_{y,i} = \sum_{j=1}^{n}y_{j}^{2}y_{j+i}^{2}, \end{equation} where indices are added modulo $n$ and this time we cannot assume that all $|x_{i}|,|y_{i}|$ are small.
Fortunately we have:
\begin{equation}
\sum_{i=1}^{n} \alpha_{x,i} = 1
\end{equation}
and 
\begin{equation}
\sum_{i=1}^{n} \alpha_{y,i} = 1
\end{equation}

Let us fix some positive function $f(k)$. We can conclude that the number of variables $\alpha_{x,i}$
such that $\alpha_{x,i} \geq \frac{f(k)}{{k \choose 2}}$ is at most $\frac{{k \choose 2}}{f(k)}$.
Notice that each such $\alpha_{x,i}$ and each such $\alpha_{y,i}$ corresponds to a pair $\{i_{1},_{2}\}$ of rows of the matrix $\mathcal{P}$ and consequently to the unique element $Cov(U_{i_{1}},U_{i_{2}})$ of the entire covariance sum (scaled by $\frac{1}{k^{2}}$).
Since trivially we have $|Cov(U_{i_{1}},U_{i_{2}})|=O(1)$, we conclude that the contribution of these elements to the entire covariance sum is of order $\frac{1}{f(k)}$.
Let us now consider these $\alpha_{x,i}$ and $\alpha_{y,i}$ that are at most $\frac{f(k)}{{k \choose 2}}$.
These sums are small (if we take $f(k)=o(k^{2})$) and thus it makes sense to apply Lemma \ref{azuma_general} to them. That gives us upper bound $a=\Delta$ with probability: 
\begin{equation}\mathbb{P}^{*} \geq 1-e^{-\Omega(a^{2}\frac{k^{2}}{f(k)})}.\end{equation}
Taking  $f(k)=(\frac{k^{2}}{\log(k)})^{\frac{1}{3}}$ and $a=\Delta = \frac{1}{f(k)}$, we conclude that:
\begin{equation}
Var(\tilde{\theta}^{n}_{p,r}) \leq \frac{1}{k}\frac{\theta(\Pi-\theta)}{\Pi^{2}} + (\frac{\log(k)}{k^{2}})^{\frac{1}{3}}
\end{equation}
Thus, from the Chebyshev's inequality, we get the following for every $c>0$ and fixed points $p,r$:
\begin{equation}
\mathbb{P}(|\tilde{\theta}^{n}_{p,r} - \frac{\theta}{\Pi}| \geq c (\frac{\sqrt{\log(k)}}{k})^{\frac{1}{3}}) = O(\frac{1}{c^{2}}).
\end{equation}
That completes the proof.

\end{document}